\documentclass[10pt,twocolumn,
amssymb, nobibnotes, notitlepage,nofootinbib aps,superscriptaddress
]{revtex4}

\usepackage{amsmath}
\usepackage{amssymb}
\usepackage{amsthm}
\usepackage[pdftex]{color}
\usepackage[pdftex,colorlinks,citecolor=blue,linkcolor=blue,urlcolor=blue]{hyperref} 
\usepackage{graphicx}
\usepackage{dcolumn} 
\usepackage{bm} 

\usepackage{longtable}

\usepackage{ulem}   
\usepackage{comment} 
\usepackage{datetime}

\usepackage{tabularx}
\usepackage{float}
\restylefloat{table}

\newcommand{\beq}{\begin{equation}}
\newcommand{\eeq}{\end{equation}}
\newcommand{\bea}{\begin{eqnarray}}
\newcommand{\eea}{\end{eqnarray}}
\newcommand{\barr}{\begin{array}}
\newcommand{\earr}{\end{array}}

\long\def\begincomment#1\endcomment{}

\newcommand{\g}{\gamma}

\newcommand{\wi}{\widetilde}

\newtheorem{remark}{Remark}

\newtheorem{proposition}{Proposition}

\begin{document}

%\sffamily
%{\color{blue}  International Journal of Theoretical Physics
%}
\title{Production of Dirac particle in  a deformed Minkowsky space-time
}

\author{ Dine Ousmane Samary}\email{dsamary@perimeterinstitute.ca}
\affiliation{Perimeter Institute for Theoretical Physics, Waterloo, ON, N2L 2Y5, Canada}
\affiliation{International Chair in Mathematical Physics and Applications 072B.P.50, Cotonou, Benin} 
\affiliation{Facult\'e des Sciences et Techniques, University of Abomey-Calavi, Benin}    \author{S\^ecloka Lazare Guedezounme}\email{guesel10@yahoo.fr} 
\affiliation{International Chair in Mathematical Physics and Applications 072B.P.50, Cotonou, Benin}  \author{Antonin Kanfon }\email{kanfon@yahoo.fr}  
\affiliation{Facult\'e des Sciences et Techniques, University of Abomey-Calavi, Benin}

\date{\today}
%\date{\currenttime, \today}

\begin{abstract}
In this paper we  study the Dirac field theory  interacting  with external gravitation field, described  with    tetrad of the form $e_b^\mu(x)=\varepsilon(\delta_b^\mu+\omega_{ba}^\mu x^a)$, where $\varepsilon=1$ for $\mu=0$ and $\varepsilon=i$ for $\mu=1,2,3.$  
The probability density of the vacuum-vacuum pair creation is given. In particular case of vanishing electromagnetic fields, we point out how this deformation    modify  the amplitude transition. The corresponding Dirac equation is solved. 
\end{abstract}


\maketitle
%\tableofcontents

\section{Introduction}
%There has been, not only recently, very much mathematical interest in Dirac theory on curve space-time.
 The Dirac particle theory arise in the theoretical  description of the fermion particles phenomena. They also  become central to elementary particle physics, as the starting point for quantum theory of electromagnetic interaction.  In the past few years this theory  have been widely studied in various curved backgrounds due to its importance in both astrophysics and cosmology, as well as in the study of particle creation processes \cite{Gavrilov:1996pz}-\cite{Samary:2014eja}. Pair production is a phenomenon of nature where energy is converted to mass. nevertheless there are only few problems for which the Dirac equation can be solved exactly. Some of them are give in \cite{Shishkin:1992js}-\cite{Hounkonnou:1999ym} and references therein. 
The explicit solution  is crucial in the particle creation processes and  is the base of the standard cosmological model 
\cite{Hack:2015zwa}-\cite{Benini:2013fia}.  This amounts to claim that, the formulation and behaviour of fermion particles physics including the gravitation field is performed using solution of Dirac equation.

In the present work, the transition amplitude and the probability density of the pair creation of the Dirac particle is examined in the twisted Minkowsky space-time. A solution of the Dirac equation is proposed.  The metric is chosen  to be  the first-order fluctuation of the flat Minkowsky pseudo-metric tensor $\eta^{ab}$.     We pointed out how this deformation of flat metric  modify  the well know probability density of creation of Dirac particle,  a while ago computed in \cite{Lin:1998rn},  \cite{Brezin:1970xf} and references therein.

The paper is organized as follows. In section \eqref{sec2}, we quickly
review Dirac particle theory, interacting with  gravity field.  In \eqref{sec3}
we compute the probability density of pair production. The case of vanishing electromagnetic (EM) fields is examined. The solution of the corresponding Dirac equation is proposed in the section \eqref{secnew}. In the last section \eqref{sec4}, we conclude our work and make some remarks.

\section{Dirac equation coupled with a weak gravitation  field}\label{sec2}
In the curve space-time the conventional affine connection $\nabla_\mu$ is replaced by the spin connection $\Gamma_\mu$ which is expressed in terms of the vierbein fields $(e_a^\mu(x))$ (see \cite{Arminjon:2014mza}-\cite{Griffiths:1979ma}).  Then any curved space  description of physics can be replaced by an equivalent and simpler flat space physics, through the vierbein transformation. So there was an equivalent formulation of general relativity involving the dynamics of the so-called spin-connection. This approach came to be known as Einstein-Cartan theory and  leads to consider  general relativity  as gauge theory approach of gravity \cite{Yepez:2011bw}.  

 An arbitrary geometrical object defined on the Riemann space-time manifold can be locally projected on the tangent Minkowski space, simply by contracting its curved indices with the vierbein and its inverse.
   For  rank $n$ tensor  object $T$, we can write
\bea
 T^{a_1a_2\cdots a_n}=e_{\mu_1}^{a_1}(x) e_{\mu_2}^{a_2}(x)\cdots e_{\mu_n}^{a_n}(x) \wi T^{\mu_1\mu_2\cdots \mu_n},
\\
T_{a_1a_2\cdots a_n}=e_{a_1}^{\mu_1}(x) e_{a_2}^{\mu_2}(x)\cdots e_{a_n}^{\mu_n}(x) \wi T_{\mu_1\mu_2\cdots \mu_n},
\eea 
where the Latin indices ($a,b,c,\cdots$) is used only for the flat space-time  and the Greek indices ($\alpha,\beta,\mu,\cdots$)   for the curve space-time. The {\it ``tilde notation''} is used only for the curve space variables. $e_\mu^a(x)$ represent the inverse of $e_a^\mu(x)$.
The metric  tensors  ${\rm g}^{\mu\nu}$ and    $\eta^{ab}={\rm diag}(1,-1,-1,-1)$ are related  by $
{\rm  g}^{\mu\nu}(x)= e_a^\mu(x) e_b^\nu(x)\eta^{ab},$ or $\eta^{ab}=e^a_\mu(x) e^b_\nu(x){\rm  g}^{\mu\nu}(x).$
The connection $\Gamma_\mu$ is
\bea
\Gamma_\mu=:\frac{1}{4}{\rm  g}_{\alpha\beta}\Big(\frac{\partial e_{\nu}^a}{\partial x^\mu} e_a^\beta-\Gamma_{\nu\mu}^\beta\Big) \wi\sigma^{\alpha\nu},\\ 
\mbox{where }\,\wi\sigma^{\mu\nu}=\frac{1}{2}[\wi\g^\mu(x),\wi\g^\nu(x)]\nonumber.
\eea

We consider the Dirac equation coupled with both gravitational and EM fields, given by the following relation
\bea\label{eqdirac}
(i\wi\g^\mu(x) D_\mu-m)\psi(x)=0,
\eea
where $ D_\mu=\partial_\mu-\Gamma_\mu+iA_\mu.$
In this expression, the vectors $\Gamma_\mu$ and  $A_\mu$ are respectively   the gravitation  and   EM gauge vectors.
The  field  $\psi$,
is a four-components
complex functions of space-time coordinates  $x^\mu,\,\,\mu =
0, 1, 2, 3$. 
The Dirac gamma matrices $\wi\g^\mu(x)$,
acting on the vector fields $\psi$, satisfy the anti commutation formula:
\bea
 \{\wi\g^\mu(x), \wi\g^\nu(x)\}=2{\rm g }^{\mu\nu}\,\mbox{ such that }\,\wi\g^\mu(x)=e_a^\mu(x) \g^a.
\eea
 The flat space-time gamma matrices $\g$ are expressed with the Pauli matrices $\sigma^{i}, \,\, i=1,2,3$ by 
\bea\label{lazonno}
\g^0=\left(\begin{array}{cc}
1_2&0\\
0&-1_2
\end{array}\right),\,
\g^i=\left(\begin{array}{cc}
0&\sigma^i\\
-\sigma^i&0
\end{array}\right),\,\, i=1,2,3.
\eea

Remark that the equation \eqref{eqdirac}  provided from the Euler-Lagrange equation of motion of  the action $S$:
\bea\label{actionnn}
S&=&\int\, d^4x \sqrt{-{\rm g}}\Big(i\bar\psi\wi\g^\mu(x) D_\mu\psi-m\bar\psi\psi\Big),
\eea
where $\bar\psi=\psi^\dag\wi\g^0(x)$.

The dynamics  described by the relation \eqref{eqdirac}  is invariant under an external  local transformation of Lorentz group. In the case of internal local transformation  this invariance is satisfy for the dynamics occurring within the space-time manifold.
The symmetries of this internal space are chosen to be the gauge symmetries of some gauge theory, so a unified theory would contain gravity together with
the other observed fields.

For $|\omega_a^\mu|<<1$, we consider the vierbein field $e_a^\mu(x)$  as
\bea\label{vie}
(e^\mu_a)(x)=&\mbox{diag}\Big[1+\omega_a^0 x^a\,,i(1+\omega_a^1 x^a)\,,\cr &i(1+\omega_a^2 x^a), i(1+\omega_a^3x^a) \Big].
\eea 
 We shall use the notation $\omega_a^\mu x^a=:\omega_{ba}^\mu x^a$. Then, the metric tensor  ${\rm\wi g}^{\mu\nu}=:\eta^{\mu\nu}+f^{\mu\nu}$ (where $f^{\mu\nu}$ is the perturbation tensor), takes the form
\bea\label{metric}
({\rm \wi g}^{\mu\nu})=&&\mbox{diag}\Big[1+2\omega_a^0 x^a,-1-2\omega_a^1 x^a,\cr
&&-1-2\omega_a^2 x^a,-1-2\omega_a^3 x^a\Big],
\eea
such that  the limit where  $(\omega)\rightarrow 0$   restore  the Minkowsky pseudo-metric. ${\rm \wi g}^{\mu\nu}$ can be considered as the first-order fluctuation of the flat Minkowsky pseudo-metric.  Now let us  choose the tensor $(\omega)$ such that  the metric depend only on the coordinates $(t=x^0,x=x^1)$, i.e.
\bea\label{lazaro}
\omega^\mu_2=\omega^\mu_3=0,\,\,\,\,
  \omega_0^\mu=\omega, \,\,\,\,\omega_1^\mu=\widetilde\omega.
\eea
 The vector $\Gamma_{\mu}$ is 
\bea\label{lazara}
(\Gamma_\mu)=\left(\begin{array}{cccc}
\Gamma_0\\ \Gamma_1\\ \Gamma_2\\ \Gamma_3
\end{array}\right) =\left(\begin{array}{cccc}
\frac{i\wi\omega}{2}\g^0\g^1\\
\frac{i\omega}{2}\g^0\g^1\\ 
\frac{i\omega}{2}\g^0\g^2-\frac{\wi\omega}{2}
\g^1\g^2\\
 \frac{i\omega}{2}\g^0\g^3-\frac{\wi\omega}{2}\g^1\g^3
\end{array}\right).
\eea
We get the following result:
\begin{proposition}
The Dirac equation in the curve background defined with the metric \eqref{metric} and coupled with EM fields   $A_\mu=(0,0,Bx,-Et)$ is given by
\bea\label{ddd1}
&&\Big[i\g^0 \partial_0-\g^j \partial_j-i\g^2 Bx+i\g^3 Et+9i\omega\g^0-3\wi\omega\g^1\cr
&&-m(1-\omega t- \wi \omega x)\Big]\psi(x^\mu)=0,\,\, j=1,2,3.
\eea
 The solution of the corresponding  equation can be split into
\bea
\psi(t,x,y,z)={\psi}(t,x) \exp\big[i(k_2y+k_3z)\big],
\eea
where $\psi(t,x)$ is a function which depends only on $t$ and $x$.
\end{proposition}
\begin{proof}
Using the relation \eqref{vie}, the Christoffel tensors   are:
\bea
&&\Gamma_{aa}^a=-\omega_a^a,\quad \Gamma_{ab}^b=-\omega_a^b, \cr 
&&\Gamma_{aa}^b=\eta_{aa}\eta_{bb}\omega_b^a,\quad \Gamma_{ab}^c=0,\quad a\neq b\neq c,
\eea
where the Einstein summation are not taking into account in the above relations.
Also the components of the Lorentz connection are
\bea\label{laz}
\Gamma_0=&&\frac{i}{2}\Big(\omega_1^0\sigma^{01}+\omega_2^0 \sigma^{02}+\omega_3^0\sigma^{03}\Big),\\
\Gamma_1=&&-\frac{1}{2}\Big(i\omega_0^1\sigma^{10}-\omega_2^1 \sigma^{12}-\omega_3^1\sigma^{13}\Big),\\
\Gamma_2=&&-\frac{1}{2}\Big(i\omega_0^2\sigma^{20}-\omega_1^2 \sigma^{21}-\omega_3^2\sigma^{23}\Big),\\
 \Gamma_3=&&-\frac{1}{2}\Big(i\omega_0^3\sigma^{30}-\omega_1^3 \sigma^{31}-\omega_2^3\sigma^{32}\Big),
\eea
which are reduced to \eqref{lazara} using \eqref{lazaro}, and 
$
\sigma^{ab}=\frac{1}{2}[\g^a,\g^b].
$
 Now, by replacing the expressions  \eqref{laz} in \eqref{eqdirac}, the Dirac equation becomes
\beq\label{exp}
\Big(i\wi\g^\mu(x)\partial_\mu-\wi\g^\mu(x) A_\mu+\omega_a\g^a-m\Big)\psi(x^\mu)=0,
\eeq
where
$
\omega_0=i(\omega_0^1+\omega_0^2
+\omega_0^3)=3i\omega,\,
\omega_1=(\omega_1^0-\omega_1^2-\omega_1^3)=-\wi\omega,\,
\omega_2=(\omega_2^0-\omega_2^1-\omega_2^3)=0,\,
\omega_3=(\omega_3^0-\omega_3^1-\omega_3^2)=0\nonumber.
$
 We choose  the external electromagnetic field as ${\bf E}=E {\bf e}_x$, ${\bf B}=B{\bf e}_x$, where ${\bf e}_x$ is the unit vector in $x$ direction.  One solution of the Maxwell equation is then $A_\mu=(0,0,Bx,-Et)$. Finally, the relation \eqref{ddd1} is well satisfy. 
 \end{proof}
\section{Transition amplitude of the model}\label{sec3}
In this section we study, how this new  metric   modify the pair creation of fermion particles.
We consider the Hilbert space of coordinates vectors $\mathcal H_{\bf x}$ such that the space-time coordinates $x^\mu=(x^0,x^1,x^2,x^3)=(t, x,y,z)$ are eigenvalue of coordinate operators $X^\mu=(t 1_4, X,Y,Z)$ acting on $\mathcal H_{\bf x}$, i.e. for $|t,x,y,z>\in \mathcal H_{\bf x}$
\bea
X^\mu|t,x,y,z>=x^\mu|t,x,y,z>.
\eea
 The Hilbert space of momentum space vectors $\mathcal H_{\bf p}$  is define as  the Fourier transformation of $\mathcal H_{\bf x}$. The  momentum  operator   $P_\mu=(P_0,P_1,P_2,P_3 ),\,$  ($P_0=i\partial_0$, $P_j=-i\partial_j,\,\, j=1,2,3$), is defined by
\bea
  P_\mu |k_0, k_1,k_2,k_3> = k_\mu |k_0, k_1,k_2,k_3>,\cr
  <k_0, k_1,k_2,k_3|t,x,y,z> = \frac{e^{ik_\mu x^\mu}}{(2\pi N)^2},\quad N\in\mathbb{R}.
\eea 
Also the curve
space-time coordinates operators are given by $\wi X^\mu=(\wi t, \wi X, \wi Y, \wi Z)$ and conjugate momentum operators $\wi P_\mu=(\wi P_0, \wi P_1,\wi P_2, \wi P_3)$ such that $\wi P_0=i{\rm g}_{00}\partial_0$ and $\wi P_j=i{\rm g}_{jj}\partial_j,\,\, j=1,2,3$.

In the path integral point of view, the action \eqref{actionnn} gives the transition amplitude of the model (or the partition function $Z(A,\Gamma)=\mathcal N\int D\psi D\bar\psi\, e^{iS}$) which is explicitly written as :
\beq
Z(A,\Gamma)=\exp\Big[-{\rm Tr}\ln\frac{i\g^\mu\partial_\mu-m+i\epsilon}{\mathcal M}\Big],
\eeq
with $
\mathcal M = i\g^\mu(\partial_\mu-\Gamma_\mu+iA_\mu)-m+i\epsilon $ and the normalization constant is defined such that $Z(0,0)=1$. Using the followings identities:
$C\wi\g_a C^{-1}=-\wi\g_a^t,\,\,$,  $(\g^\mu)^t =-e_a^\mu C\wi\g^a C^{-1}=-C\g^\mu C^{-1},$  $ \Gamma_\mu^t=-C\Gamma_\mu C^{-1},$ 
we come to
$\mathcal M^t={iC\g^\mu(\partial_\mu-\Gamma_\mu+iA_\mu)C^{-1}-m+i\epsilon}$ and then the conjugate of the functional
$Z(A,\Gamma)$ is given by
\beq
Z^t(A,\Gamma)=\exp\Big[-{\rm Tr}\ln\frac{iC\g^\mu C^{-1}\partial_\mu+m+i\epsilon}{\mathcal M^t}\Big],
\eeq
 where  $C=i\wi\g^2\wi\g^0$.

We now compute the transition amplitude $|Z(A,\Gamma)|^2$. For this, let us define the quantities $\mathcal X_H(\omega,0)=\omega\mathcal X^1_H$, $\mathcal X_H(\omega,0)=\omega\mathcal X^1_H $,
$\mathcal Y_H(\omega,0,E,B)=\omega\mathcal Y^1_H$,
$\mathcal Y_H(0,\wi\omega,E,B)=\wi\omega\mathcal Y^2_H$, such that
\bea
\mathcal X^1_H 
= 2t|{\bf P}|^2
+\g^0\g^1 P_1+\g^0\g^2 P_2+\g^0\g^3 P_3,
\eea
\bea
\mathcal X^2_H =2X|{\bf P}|^2+\g^0\g^1 P_0+i\g^1\g^2P_2+i\g^1\g^3 P_3,
\eea
\bea
\mathcal Y^1_H
&=&2t|{\bf P}|^2+ 4\g^0\g^1 P_1+(4\g^0\g^2+4BXt)P_2\cr
&+&(4\g^0\g^3-4Et^2)P_3-6\g^0\g^3Et+4\g^0\g^2 BX\cr
&+& 2i\g^1\g^2 Bt
+2t(B^2 X^2 +E^2 t^2),
\eea
\bea
\mathcal Y^2_H&=&2X|{\bf P}|^2-2\g^1\g^0 P_0+(2i\g^1\g^2+4BX^2)P_2\cr
&+&(2i\g^1\g^3 -4EXt)P_3-i\g^1\g^3 Et+3i\g^1\g^2 BX\cr
&-&2\g^0\g^3 EX
+2X(B^2X^2+E^2t^2).
\eea
  Then 
$\mathcal P=:|Z(A,\Gamma)|^2$ is explicitly written as
\bea
&&\mathcal P=\exp\Big[-{\rm Tr}\ln\frac{\ell(\omega,\wi\omega)}{n(\omega,\wi\omega,E,B)}\Big]\cr
&&=\exp\Big[-{\rm Tr}\int_0^\infty\frac{ds}{s}\Big(e^{is n(\omega,\wi\omega,E,B)}
-e^{is\ell(\omega,\wi\omega)}\Big)\Big]\cr
&&
\eea
where
\beq
\ell (\omega,\wi\omega)=\mathcal X_H(0,0)+\mathcal X_H(\omega,0)+\mathcal X_H(0,\wi\omega),
\eeq
\bea
n(\omega,\wi\omega,E,B)&=&\mathcal Y_H(0,0,E,B)+\mathcal Y_H(\omega,0,E,B)\cr
&&+\mathcal Y_H(0,\wi\omega,E,B).
\eea 
  We get the following statement:
\begin{proposition}\label{central}
Consider that the EM fields are vanishing.  For very small positif parameter $\epsilon$ of the  size  $1/\wi\omega^2$, the probability of the pair production takes the form
\beq
\mathcal P=\exp\Bigg\{-\frac{\pi M e^\pi m^8}{1024 v \epsilon}\Big[2 N_\gamma-\frac{37}{12} - \frac{1}{3} \big(4 \ln b - \ln a\big) \Big] \Bigg\}, 
\eeq
where $\omega=v\wi\omega$, $a = \frac{3 \wi \omega}{8 v} (v^2 + 1) $,  $ b = \frac{3 \wi \omega}{2 v} (4 v^2 + 1) $, $N_\gamma$ is the Euler number, $M= \int \, \mu (y,z)dy\,dz,$ and  $\mu(y,z)$ is the test function.
\end{proposition}
\begin{remark}
 Note that the case where $M>0$ is not fulfils. This leads to a infinite probability density.   We choose the test function $\mu(y,z)$ such that  $M<0$,  and then
$$
\frac{\pi M e^\pi m^8}{1024 v }\Big[2 N_\gamma-\frac{37}{12} - \frac{1}{3} \big(4 \ln b - \ln a\big) \Big]>0.
$$
In the figure \eqref{fig:plot1} we give the plot of  this probability density as function of the parameter $\epsilon$. This figure gives asymptotically  the values of $f(\epsilon)=\mathcal P$ when $\epsilon$ tends to zero.
\begin{figure}[htbp]
\begin{center}
\includegraphics[scale=0.80]{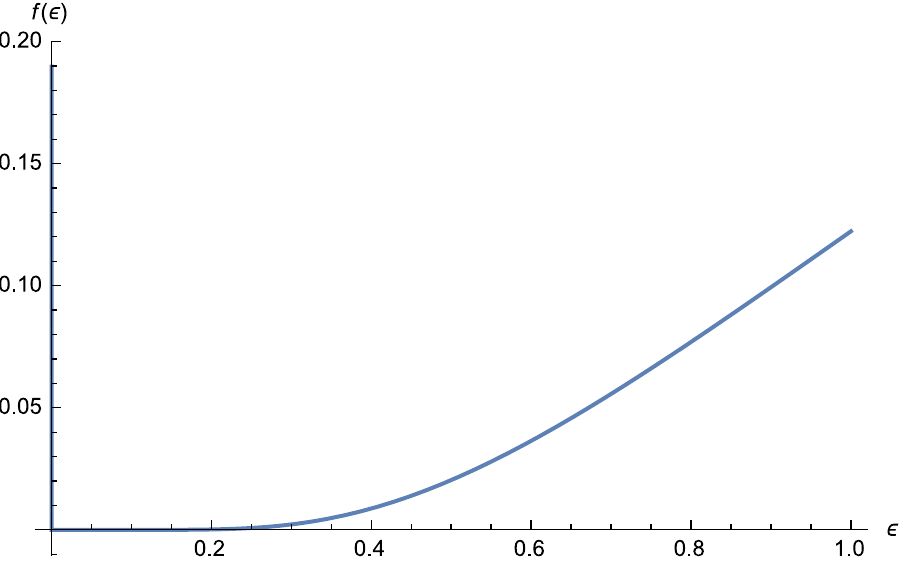}\hspace{0.5cm}
\end{center}
 \caption{Plot of $\mathcal P= f(\epsilon)$ with $v=1$, $M=-2\pi$, $m=1$, $\omega=\sqrt{\epsilon}$. }
  \label{fig:plot1} 
\end{figure}
Then the limit $\epsilon \rightarrow 0$ leads to $\mathcal P\approx 0$.
\end{remark}
\begin{proof} of the proposition \eqref{central}:
  The rest of this section is devoted to the proof of the proposition \eqref{central}.
The density $\Omega=:\Omega(\omega,\wi\omega,E,B)$ such that $ 
\mathcal P= \exp(-\Omega)
$ can be expanded as
%\begin{widetext}
\bea
\Omega 
&=&\Omega_1(\omega,\wi\omega,E,B)-\Omega_2(\omega,\wi\omega,E,B),
\eea
where
\bea
\Omega_1= \int d{\bf x}\int_0^\infty\frac{ds}{s}\langle {\bf x}| e^{is \mathcal Y_H(\omega,\wi\omega,E,B)}|{\bf x}\rangle,
\eea
\bea
\Omega_2=\int d{\bf x}\int_0^\infty\frac{ds}{s}\langle {\bf x}|e^{is\mathcal X_H(\omega,\wi\omega)}|{\bf x}\rangle.
\eea
%\end{widetext}
 We  consider the mean values:
\bea
x^1_H&=&<{\bf x}|\mathcal X^1_H|{\bf x}>,\quad
 x^2_H=<{\bf x}|\mathcal X^2_H|{\bf x}>,\cr
y^1_H&=&<{\bf x}|\mathcal Y^1_H|{\bf x}>,\quad y^{2}_H=<{\bf x}|\mathcal Y^2_H|{\bf x}>.
\eea
Now we choose  $\omega=v\wi\omega$ and $vt+x>0$ and we are  focussing on the computation of  $\Omega_1(\omega,\wi\omega,E,B)$.  We get
\bea
\Omega_1
&=&\int_0^\infty\frac{ds}{s}\langle {\bf x}| e^{is \mathcal Y_H(0,0,E,B)}|{\bf x}\rangle\cr
&\times& \int d{\bf x} d{\bf k}\,\Big(1+is\omega y^1_H+is\wi\omega  y^2_H\Big)\qquad
\eea
with
\beq\label{sat2}
\langle {\bf x}| e^{is \mathcal Y_H}|{\bf x}\rangle=\frac{-iEB \coth(Es)\cot(Bs) }{4\pi^2}e^{-ism^2}.
\eeq
 A simple routine checking shows that 
\bea\label{bel1}
&&\int d{\bf x} d{\bf k}\,\exp\Big[ is\wi\omega\Big(v y^1_H+  y^{2}_H\Big)\Big]\cr
&&=\frac{M\pi^2 e^\pi}{4s^2\epsilon}\int\,\frac{dx dt}{(vt +x)^2} \Big\{\cr
&& \exp\Big[is\wi\omega\Big(-\frac{P(x,t)}{vt+x} + Q(x,t)\Big)\Big]\Big\}
\eea
with
\bea
P(x,t) &=& 2v^2(3+2Bxt \g^0\g^2+ B^2 x^2 t^2 \cr
& - & 2\g^0\g^3Et^2+E^2 t^4) + 2v(2Bx^2\g^0\g^2 \cr
& + & i\g^1\g^2 Bxt +2 B^2 x^3 t-2\g^0\g^3 Ext \cr
& - &i\g^1\g^3Et^2+E^2xt^3)+ \frac{3}{2}+2i\g^1\g^2 Bx^2 \cr
& + & 2 B^2 x^4 - 2i\g^1\g^3 Ext+2E^2x^2 t^2,
\eea
and
\bea
Q(x,t)&=&v\big[4\g^0\g^2 Bx-6\g^0\g^3Et
+2i\g^1\g^2 Bt\cr
&+&2t(B^2 x^2 +E^2 t^2)\big]
-i\g^1\g^3 Et\cr
&+& 3i\g^1\g^2 Bx
- 2\g^0\g^3 Ex \cr
&+&2x(B^2x^2+E^2t^2).
\eea
In the same manner $\Omega_2$ takes the form
\bea
\Omega_2&=&\int_0^\infty\frac{ds}{s}\int d{\bf x} d{\bf k}\,\langle {\bf x}| e^{is \mathcal X_H(0,0)}|{\bf x}\rangle  \cr
&\times& \Big(1+is\omega x^1_H+is\wi\omega  x^2_H\Big)
\eea
where
\bea
\langle {\bf x}| e^{is \mathcal X_H}|{\bf x}\rangle=-\frac{i}{16\pi^2 s^2}e^{-ism^2}.
\eea
We can  now show that
\bea\label{bel2}
&&\int d{\bf x} d{\bf k}\,\Big(1+isv\wi\omega x^1_H+is\wi\omega  x^{2}_H\Big)\cr
&& =\int\,d{\bf x} \frac{\pi^2 e^\pi}{4s^2\epsilon(vt+x)^2} \exp\Big[-\frac{3i s\wi\omega(v^2+1)}{8(vt+x)}\Big].\cr
&&
\eea
However
\bea\label{sauvage}
&&\Omega(\wi\omega,E,B)=\frac{iM e^\pi}{64\epsilon}\int_0^\infty\,ds\, \frac{ e^{-ism^2}}{s^3}\Big[\frac{1}{s^2}\mathcal I(t_0,x_0)\cr
&&-4EB\coth(Es)\cot(Bs)\mathcal J(t_0,x_0)\Big]
\eea
where $M$ is chosen to be
$
 M = \int \,\mu(x,y) dy\,dz<\infty, 
$
\bea
\mathcal I(t_0,x_0)&=&\int_{x_0}^\infty \int_{t_0}^\infty\,dx dt\, \frac{\pi^2 e^\pi}{4s^2\epsilon(vt+x)^2}\Big\{\cr
&&\exp\Big[-\frac{3i s\wi\omega(v^2+1)}{8(vt+x)}\Big]\Big\}
\eea
and
\bea
\mathcal J(t_0,x_0)=\int_{x_0}^\infty \int_{t_0}^\infty\,dx dt\,\frac{1}{(vt +x)^2}\cr
\times \exp\Big[is\wi\omega\Big(-\frac{P(x,t)}{(vt+x)} + Q(x,t)\Big)\Big].
\eea
The  integral \eqref{sauvage}  exhibit the divergence at point $x=t=0$. This shall be regularized by using the Cauchy principal value.  For $E=B=0$ we get 
\bea\label{xxx1}
\mathcal I(0,0)&=&\frac{1}{\alpha_1}\int_0^\infty\, dt\, \Big[\sin(\alpha_1/vt)+2i\sin^2(\alpha_1/2vt)\Big]\cr
&=&\frac{i\pi}{2v}+\frac{1}{v}\Big[1-N_\gamma-\ln\Big(\frac{\alpha_1}{v}\Big)\Big], \\ 
\label{xxx2}\mathcal J(0,0)&=&\frac{1}{\alpha_2}\int_0^\infty\, dt\, \Big[\sin(\alpha_2/vt)+2i\sin^2(\alpha_2/2vt)\Big]\cr
&=&\frac{i\pi}{2v}+\frac{1}{v}\Big[1-N_\gamma-\ln\Big(\frac{\alpha_2}{v}\Big)\Big], 
\eea
with $\alpha_1=\frac{3s\wi\omega}{8}(v^2+1)$, $\alpha_2=\frac{s\wi\omega}{2}(12v^2+3)$ and $N_\gamma$ is the Euler number given by $N_\gamma=0.577215664$. Also, for $a, b\in \mathbb{R}$ the integral  
$$
\mathcal Q=\int_0^\infty\,ds\, \frac{ e^{-ism^2}}{s^5}\Big(4\ln(bs)-\ln(as)\Big)
$$
admits the Cauchy principal value 
\bea\label{xxx3}
Pv( \mathcal Q)&=&\frac{m^8}{1152} \Big[415 - 300 N_\gamma + 72N_\gamma^2 + 12 \pi^2 \cr
    &+& 96 \ln b^2  - 
    24 \ln a^2- 300 \ln (i m^2) \cr 
    &+& 72 \ln (i m^2)^2 + 144 N_\gamma \ln (i m^2) \cr
&+& 
    4 \Big(-25 + 12 N_\gamma + 12 \ln (i m^2)\Big)\ln a \cr  
    &-&16 \Big(-25 + 12 N_\gamma + 12 \ln (i m^2)\Big)\ln b \Big].\cr
&&
\eea
Remark that the probability density of pair creation in the limit $\wi\omega=0$ (see \cite{Lin:1998rn}) correspond to 
\beq
\Omega(0,E,B)=\frac{EB}{4\pi^2}\sum_{k=1}^\infty\, \frac{1}{k}\coth\Big(k\pi\frac{B}{E}\Big)\exp\Big(-\frac{k\pi m^2}{E}\Big).
\eeq
Using the Taylor expansion  as
\beq\label{sat}
\Omega(\wi\omega,E,B)
=\Omega(0,E,B)+\wi\omega\Omega'(0,E,B)+ \mathcal O(\wi\omega^2),
\eeq
 we come to $\Omega(\wi\omega,0,0)
=\wi\omega\Omega'(0,0,0)+ \mathcal O(\wi\omega^2)$ and 
\bea\label{lll}
\Omega(\wi\omega,0,0)=\frac{iM e^\pi}{64\epsilon}\int_0^\infty\,ds\, \frac{ e^{-ism^2}}{s^5} \mathcal T(0,0),\\
\mathcal T(0,0)=\mathcal I(0,0)-4\mathcal J(0,0)\nonumber.
\eea
 We choose the real part of  $\Omega(\wi\omega,0,0)$ denoted by $\Re_e \Omega(\wi\omega,0,0)=[\Omega(\wi\omega,0,0)
+\Omega^*(\wi\omega,0,0)]/2$. Using \eqref{xxx1}, \eqref{xxx2}  and \eqref{xxx3}  
\bea
\Re_e\Omega(\wi\omega,0,0)&=&-\frac{\pi M e^\pi m^8}{1024 v \epsilon}\Big[2 N_\gamma-\frac{37}{12} \cr
&-& \frac{1}{3} \Big(4 \ln b - \ln a\Big) \Big].
\eea
where 
$a = \frac{3 \wi \omega}{8 v}(v^2 + 1)$, 
 $ b = \frac{3 \wi \omega}{2 v} (4 v^2 + 1) $.
Finally  it is straightforward to check the following relation
\bea
\wi\omega\Omega'(0,0,0)&=&\frac{\pi M e^\pi m^8}{1024 v \epsilon}\Big[2 N_\gamma-\frac{37}{12}\cr
 &-& \frac{1}{3} \Big(4 \ln b - \ln a\Big) \Big].
\eea
While the probability of the pair production takes the form
\bea
\mathcal P&=&\exp\Big\{ -\frac{\pi M e^\pi m^8}{1024 v \epsilon}\Big[2 N_\gamma-\frac{37}{12}\cr
&-& \frac{1}{3} \Big(4 \ln b - \ln a\Big) \Big] \Big\} \approx 0. 
\eea
This end the proof of proposition \eqref{central}.
\end{proof} 

\section{Solution of the Dirac equation}\label{secnew}
In this section we give the solution of the Dirac equation \eqref{ddd1}.  We consider the operators $\mathcal K_1$ and $\mathcal K_2$ satisfying the commutation relation $[\mathcal K_1,\mathcal K_2]=0$ and   given by
\bea
\label{k1}\mathcal K_1&=&i\g^0 \partial_0-i\g^2k_2-i\g^3k_3+i\g^3 Et+9i\omega\g^0\cr
&&-3\wi\omega\g^1-m(1-\omega t),
\eea
\bea
&&\label{k2}\mathcal K_2=-\g^1\partial_1-i\g^2 Bx+m \wi \omega x.
\eea
The equation   \eqref{ddd1} takes the form
$
(\mathcal K_1+\mathcal K_2)\psi(t,x)=0
$
and admit  separate variables as $\psi(t,x)=\psi(t)\psi(x)$.  For a constant $\lambda\in\mathbb{C}$, we get the two   eigenvalue equations 
\bea
\label{eigen1}\mathcal K_1\psi(t,x)=\lambda\psi(t,x)\\
\label{eigen2}\mathcal K_2\psi(t,x)=\lambda\psi(t,x).
\eea
Consider the equation  \eqref{eigen1}.
We write the four vector $\psi(t)$  as $\psi(t)=(\psi_1(t),\psi_2(t))$ and $\psi_j(t)=(\psi_{ja}(t),\psi_{jb}(t)),\,\, j=1,2$, and the equation \eqref{eigen1}
leads to
\bea\label{ddds}
&&\big({L'L}+{D^2}+CC'\big)\psi_{1a}(t)=0,\\
&&\big({LL'}+{D^2}+CC'\big)\psi_{2a}(t)=0,\\
&&\label{dddu} C\psi_{1b}(t)={L'}\psi_{2a}(t)-{D}\psi_{1a}(t),\\
&& C\psi_{2b}(t)=-{L}\psi_{1a}(t)-{D}\psi_{2a}(t)
\eea
where
\bea
&&L=i\partial_0+9i\omega-m(1-\omega t)+\lambda,\nonumber\\
&& L'=-i\partial_0-9i\omega+m(1-\omega t)+\lambda\nonumber\\
&&C=-3\wi\omega-k_2,\quad C'=-3\wi\omega+k_2,\nonumber\\ 
&&D=iEt-ik_3.\nonumber
\eea
Let
\bea
f(t)&=&\frac{E}{2}t^2+\Big(\frac{m}{E}(m-\lambda)\omega-9\omega  -k_3 \Big)t,
\eea
\bea
 g(t)&=&(-1)^{1\over 4}(iE)^{\frac{1}{2}}t
\eea
\bea
\delta^1_E&=&\frac{m\omega}{2E^3}(m-\lambda)\cr
&+&\frac{1}{2E}\Big(k_2^2-(m-\lambda)^2+im\omega\Big)-\frac{1}{2}
\eea
\bea
\delta^2_E&=&\frac{(-1)^{1\over 4}}{(iE)^{\frac{3}{2}}}\Big(E k_3-(m-\lambda)m\omega\Big).
\eea 
The solution of the equations \eqref{ddds}   are  a linear combination of Hermite and (1,1)-hypergeometric polynomial given by
\bea\label{laz1}
\psi_{1a}(t)&=&c_1 e^{f(t)}\mathcal H\Big[\delta^{1}_E,\delta^2_E+g(t)\Big]\cr
&+&c_2 e^{f(t)} {}_1F_1[\frac{\delta^1_E}{2},\frac{1}{2},\big(\delta^2_E+g(t)\big)^2]\cr
\label{laz2}\psi_{2a}(t)&=&c_1 e^{f(t)}\mathcal H\Big[\bar\delta^{1}_E,\delta^2_E+g(t)\Big]\cr
&+&c_2 e^{f(t)} {}_1F_1[\frac{\bar\delta^{1}_E}{2},\frac{1}{2},\big(\delta^2_E+g(t)\big)^2].
\eea
the solutions of the equation \eqref{dddu}  can be simple obtained using the followings identities:
\bea
&&\frac{d\, \mathcal H(a,b+ct)}{dt}=2ac \mathcal H(-1+a, b+ct),\\ 
&&\frac{d \, {}_1F_1(a,b,ct^2+dt+e)}{dt}\cr
&&=\frac{a(d+2ct)}{b}{}_1F_1(1+a,1+b,ct^2+dt+e).
\eea

Now, consider  the equation \eqref{eigen2}.
Using the Dirac matrices \eqref{lazonno}, we get
\bea\label{tangent}
-(\sigma_1\partial_1+i\sigma_2 Bx)\psi_2(x)+(m\wi\omega x-\lambda)\psi_1(x)=0\\
\label{tangent1}(\sigma_1\partial_1+i\sigma_2 Bx)\psi_1(x)+(m\wi\omega x-\lambda)\psi_2(x)=0
\eea
where $\psi(x)=(\psi_1(x),\psi_2(x))$.    Let us define the quantities $b(B)$, $r(B)$ and $s(B)$ as
\bea
&&b(B)=:\frac{\lambda^2}{2B},\nonumber\\ &&r(B)=\frac{2m\wi\omega\lambda(-1)^{\frac{1}{4}}}{\sqrt{2}(iB)^\frac{3}{2}},\nonumber\\
  &&s(B)=(-1)^\frac{1}{4}\sqrt{2}(iB)^\frac{1}{2}\nonumber.
\eea
 For  $\psi_j(x)=(\psi_{ja}(x),\psi_{jb}(x)), \,\, j=1,2$. The solutions of the equation \eqref{tangent} are 
\bea\label{tax}
\psi_{1a}(x)&=&c_1\mathcal D\Big[-b(B),-r(B)+s(B)x\Big]
\cr
&+&c_2\mathcal \mathcal D\Big[-b(B),-ir(B)+is(B)x\Big]
\\
\label{tax1}\psi_{1b}(x)&=&c_1\mathcal D\Big[-1-b(B),-r(B)+s(B)x\Big]
\cr
&+&c_2\mathcal D\Big[-1-b(B),-ir(B)+is(B)x\Big].\cr
&&
\eea
However, the equation  \eqref{tangent1} can be split 
 into
\bea
\label{mmm}\Big[\partial_1^2-B-B^2 x^2+\lambda^2-2m\lambda\wi\omega x\Big]\psi_{1a}(x)=0\\
\label{nnn}\Big[\partial_1^2+B-B^2 x^2+\lambda^2-2m\lambda\wi\omega x\Big]\psi_{1b}(x)=0,
\eea
and the solutions are well  given by the following:
\bea
\psi_{2a}(x)=\frac{1}{\lambda-m\wi\omega x}(\partial_1+Bx)\psi_{1b},\\\psi_{2b}(x)=\frac{1}{\lambda-m\wi\omega x}(\partial_1-Bx)\psi_{1a}.
\eea
where 
 the identities
\bea
\frac{d}{dx}\mathcal D(a, bx+c)&=&\frac{b}{2}(bx+c)
\mathcal D(a, bx+c)\cr
&-&b \mathcal D(1+a, bx+c),
\eea
are usefull.

%\begin{remark}
%In the case where the parameters $\omega$ and $\wi\omega$ tend to zero the solution provided from the expressions: \eqref{tax}, \eqref{tax1}, \eqref{sol1}, \eqref{sol2}, \eqref{laz1}, \eqref{laz2}, \eqref{laz3} and \eqref{laz4} are reduced to the solution of Dirac equation in the flat Minkowsky space.
%\end{remark}

\section{Conclusion}\label{sec4}
In this paper, we have  computed the probability density of pair production of the fermion particles. The case of vanishing EM fields is scrutinized explicitly.  Hereafter  we  will shed light on the case of non-vanishing EM fields, which has  not been entirely considered in this paper. 
In the other hand,  we have solved the Dirac equation coupled with gravitation field, using the separation of variables.  The solutions are expressed in terms
of hypergeometric functions. The limit where the deformation parameter $\wi\omega$ tends to zero  is given.

\section*{Acknowledgements} 
D. O. S. research is supported in part by the Perimeter Institute for Theoretical Physics (Waterloo) and by the Fields Institute for Research in Mathematical Sciences (Toronto).
Research at the Perimeter Institute is supported by the Government of Canada through Industry Canada and by the Province of Ontario through the Ministry of Economic Development \& Innovation.

%%%%%%%%%%%%%%%%%%%%%%%
%%%%%%%%%%%%%%%%%%%%%%%%%%%%%%%%%%%%%%%%%%%%%%%%%%%%%

\end{document}